\documentclass[11pt]{article}
\usepackage{color,amsthm,amsmath,amsfonts,xspace}
\usepackage{MnSymbol}
\usepackage{fullpage}
\usepackage{hyperref}
\usepackage[toc,page]{appendix}
\usepackage{authblk}
\usepackage[absolute,overlay]{textpos}
\usepackage{tikz}
\usepackage{verbatim}

\title{
Type Size Code for Compressing Erd\"{o}s-R\'enyi Graphs \thanks{This research was conducted independently by the author.} }
\author[1]{Nematollah Iri\thanks{ nematollah.iri@asu.edu}}

\newtheorem{Theorem}{Theorem}
\newtheorem{Theorem*}{Theorem}
\newtheorem{Lemma}[Theorem]{Lemma}
\newtheorem{Lemma*}[Theorem]{Lemma}

\begin{document}

\maketitle

\begin{abstract}
We consider universal source coding of unlabeled graphs which are commonly referred to as graphical structures. We adopt an Erd\"{o}s-R\'enyi model to generate the random graphical structures. We propose a variant of the previously introduced Type Size code, where type classes are characterized based on the number of edges of the graphical structures. The proposed scheme sorts the graphical structures based on the size of their type classes and assigns binary sequences to them in this order. The $\epsilon$-coding rate of the Type Size code (up to the third-order term) for compressing graphical structures is derived.
\end{abstract}
\section{Introduction}
Many problems in social networks, world wide web, recommendation systems, biology and etc., reduce to computations and processings over graphs. With the emergence of big data, such graphs may contain trillion edges. Moreover, next generation applications in the big data era, enforce more stringent I/O access and latency requirements. Therefore, new methods for compressing such massive graphs are essential for retrieval and processing over short time scales.

We aim at compressing the underlying graph up to isomorphism, i.e. to compress the structure of the graph. Many works, such as those in the area of graph summarization \cite{yike}, consider the compressed version of the graphical structure to be a graph itself \cite{fZhou}. Other works, such as \cite{ChoiwSzG, mohri, TomasZGSZ}, follow an information theoretic paradigm, where the graph structure is mapped into binary strings. We focus on lossless graph compression to faithfully recover the original graph structure from its encoded bits. Numerous lossless graph compression algorithms have been proposed in the literature. See \cite{Besta}, for an exhaustive survey on the existing works, .

Prefix-free assumption has been traditionally imposed to unambiguously decode the block of codewords. However, for many applications in storage and retrieval, there are out-of-band markers to navigate the boundaries of the data to be compressed, hence, one may relax the prefix-free assumption for compression. Such a compression is referred to as one-to-one coding \cite{szpan}.

We follow a \emph{universal} one-to-one source coding setup, where the underlying probability distribution generating the data is arduous to estimate or unknown, yet presumed to belong to a known class of distributions. A universal one-to-one compression scheme has been first introduced in \cite{oliverLalithaJournal}--- the Type Size (TS) code--- to optimally compress the class of all stationary memoryless sources. TS code, universally \emph{orders} sequences based on the size of their type classes, and subsequently map them to binary strings in a lexicographic order. The result in \cite{oliverLalithaJournal} shows a gain of logarithmic order in the input size by relaxing the prefix-free assumption.

There is an underlying flexibility in defining type classes from a TS code perspective. In fact, the one-to-one compression problem is equivalent to characterizing the type classes that lead to optimal performance \cite{nematCISS}. Characterizing the type classes based on the empirical probability mass function of the sequences is shown to be optimal for compression of the class of \emph{all} i.i.d. \cite{oliverLalithaJournal} and Markov \cite{nemat} sources over a finite alphabet, while the quantized type classes \cite{nematKosutJournal} are shown to be optimal for compression of the parametric exponential family of distributions.

We adopt the TS approach for compressing unlabeled Erd\"{o}s-R\'enyi graphs, which we refer to as graphical \emph{structures}, with an unknown edge probability of $p$. We define two graphical structures to be in the same type class if and only if they have the same number of edges. Note that two graphical structures within a type class have the same probability regardless of the enforced Erd\"{o}s-R\'enyi model. We order the graphical structures based on the size of the type class they belong to, from smallest to largest, and subsequently map them to binary strings lexicographically (See Figure \ref{TScodeGraphs} for an example). We show that with probability at least $(1-\epsilon)$, the TS code requires at most
\begin{equation}
\label{thisEquationFirstt}
{n\choose 2} H(p) + \sigma(p)\sqrt{{n\choose 2}}Q^{-1}(\epsilon) -\log{n!} + \mathcal{O}(1)
\end{equation}
bits, to compress a graphical structure with $n$ nodes, where $H(p)$ and $\sigma^2(p)$ are the entropy and the varentropy of the underlying Erd\"{o}s-R\'enyi model $p$, respectively. The first and third-order terms in (\ref{thisEquationFirstt}) do match with an earlier result in \cite{ChoiwSzG}. However, we depart from an average case analysis in \cite{ChoiwSzG} to a probabilistic analysis.

The required memory to store the TS code ordering grows super-exponentially with the graph size \footnote{The size of a graph is the number of its nodes.}. Even though, this memory requirement may sound prohibitive, however, many applications may afford this storage requirement in order to achieve the rate optimality. A few examples include but are not limited to
\begin{itemize}
\item Long distance communications such as satellite communications,
\item Transmission over VLF and ELF frequencies which offer a very low bandwidth. Optimal-rate compression comes into awareness in such scenarios, since one can transmit only a limited few (hundreds) bits in a minute.
\end{itemize}

The rest of the paper is organized as follows. We introduce the lossless source coding of the graphical structures and related definitions in Section \ref{PrelimSection}. In Section \ref{TSCodeSection}, we describe the Type Size code and provide tight bounds on the type class sizes. In Section \ref{mainResult}, we present the main theorem of the paper which is proved in Section \ref{ThmProof}. We conclude in Section \ref{SecConclusion}.

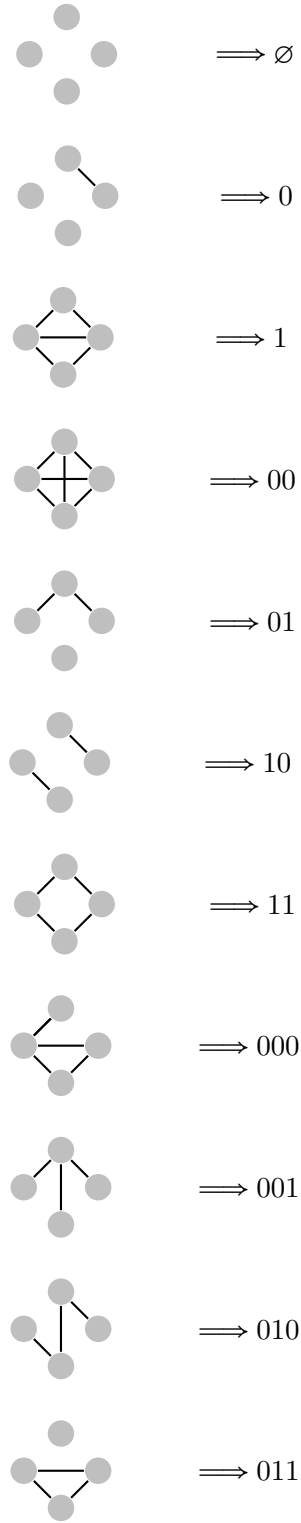
\begin{figure}[t]
\label{TScodeGraphs}
\begin{centering}

\begin{tikzpicture}[auto, node distance=0.7cm, every loop/.style={},
                    thick,main node/.style={circle,fill=black!25,minimum size=10pt,inner sep=0pt}]
  \node[main node] (1) {};
  \node[main node] (2) [below left of=1] {};
  \node[main node] (3) [below right of=2] {};
  \node[main node] (4) [below right of=1] {};
  \node at (2.5,-0.5) {$\Longrightarrow \emptyset$};
\end{tikzpicture}

\vspace{0.2in}

\begin{tikzpicture}[auto, node distance=0.7cm, every loop/.style={},
                    thick,main node/.style={circle,fill=black!25,minimum size=10pt,inner sep=0pt}]
  \node[main node] (1) {};
  \node[main node] (2) [below left of=1] {};
  \node[main node] (3) [below right of=2] {};
  \node[main node] (4) [below right of=1] {};

  \path[every node/.style={minimum size=5pt}]
    (1) edge node [left] {} (4);
  \node at (2.5,-0.5) {$\Longrightarrow 0$};
\end{tikzpicture}

\vspace{0.2in}

\begin{tikzpicture}[auto, node distance=0.7cm, every loop/.style={},
                    thick,main node/.style={circle,fill=black!25,minimum size=10pt,inner sep=0pt}]
  \node[main node] (1) {};
  \node[main node] (2) [below left of=1] {};
  \node[main node] (3) [below right of=2] {};
  \node[main node] (4) [below right of=1] {};

  \path[every node/.style={minimum size=5pt}]
    (1) edge node [left] {} (4)
    (2) edge node [right] {} (1)
        edge node [left] {} (3)
        edge node [right] {} (4)
    (3) edge node [right] {} (4);

  \node at (2.5,-0.5) {$\Longrightarrow 1$};
\end{tikzpicture}
\vspace{0.2in}

\begin{tikzpicture}[auto, node distance=0.7cm, every loop/.style={},
                    thick,main node/.style={circle,fill=black!25,minimum size=10pt,inner sep=0pt}]
  \node[main node] (1) {};
  \node[main node] (2) [below left of=1] {};
  \node[main node] (3) [below right of=2] {};
  \node[main node] (4) [below right of=1] {};

  \path[every node/.style={minimum size=5pt}]
    (1) edge node [left] {} (4)
    (2) edge node [right] {} (1)
        edge node [left] {} (3)
        edge node [right] {} (4)
    (3) edge node [right] {} (4)
        edge node [right] {} (1);;

  \node at (2.5,-0.5) {$\Longrightarrow 00$};
\end{tikzpicture}

\vspace{0.2in}
\begin{tikzpicture}[auto, node distance=0.7cm, every loop/.style={},
                    thick,main node/.style={circle,fill=black!25,minimum size=10pt,inner sep=0pt}]
  \node[main node] (1) {};
  \node[main node] (2) [below left of=1] {};
  \node[main node] (3) [below right of=2] {};
  \node[main node] (4) [below right of=1] {};

  \path[every node/.style={minimum size=5pt}]
    (1) edge node [left] {} (4)
    (2) edge node [right] {} (1);
  \node at (2.5,-0.5) {$\Longrightarrow 01$};
\end{tikzpicture}

\vspace{0.2in}

\begin{tikzpicture}[auto, node distance=0.7cm, every loop/.style={},
                    thick,main node/.style={circle,fill=black!25,minimum size=10pt,inner sep=0pt}]
  \node[main node] (1) {};
  \node[main node] (2) [below left of=1] {};
  \node[main node] (3) [below right of=2] {};
  \node[main node] (4) [below right of=1] {};

  \path[every node/.style={minimum size=5pt}]
    (1) edge node [left] {} (4)
    (2) edge node [right] {} (3);
  \node at (2.5,-0.5) {$\Longrightarrow 10$};
\end{tikzpicture}
\vspace{0.2in}

\begin{tikzpicture}[auto, node distance=0.7cm, every loop/.style={},
                    thick,main node/.style={circle,fill=black!25,minimum size=10pt,inner sep=0pt}]
  \node[main node] (1) {};
  \node[main node] (2) [below left of=1] {};
  \node[main node] (3) [below right of=2] {};
  \node[main node] (4) [below right of=1] {};

  \path[every node/.style={minimum size=5pt}]
    (1) edge node [left] {} (4)
    (2) edge node [right] {} (1)
        edge node [left] {} (3)
    (3) edge node [right] {} (4);

  \node at (2.5,-0.5) {$\Longrightarrow 11$};
\end{tikzpicture}

\vspace{0.2in}

\begin{tikzpicture}[auto, node distance=0.7cm, every loop/.style={},
                    thick,main node/.style={circle,fill=black!25,minimum size=10pt,inner sep=0pt}]
  \node[main node] (1) {};
  \node[main node] (2) [below left of=1] {};
  \node[main node] (3) [below right of=2] {};
  \node[main node] (4) [below right of=1] {};

  \path[every node/.style={minimum size=5pt}]
    (1) edge node [left] {} (2)
    (2) edge node [right] {} (1)
        edge node [left] {} (3)
        edge node [right] {} (4)
    (3) edge node [right] {} (4);

  \node at (2.5,-0.5) {$\Longrightarrow 000$};
\end{tikzpicture}

\vspace{0.2in}

\begin{tikzpicture}[auto, node distance=0.7cm, every loop/.style={},
                    thick,main node/.style={circle,fill=black!25,minimum size=10pt,inner sep=0pt}]
  \node[main node] (1) {};
  \node[main node] (2) [below left of=1] {};
  \node[main node] (3) [below right of=2] {};
  \node[main node] (4) [below right of=1] {};

  \path[every node/.style={minimum size=5pt}]
    (1) edge node [left] {} (4)
    (1) edge node [right] {} (3)
    (2) edge node [right] {} (1);
  \node at (2.5,-0.5) {$\Longrightarrow 001$};
\end{tikzpicture}

\vspace{0.2in}

\begin{tikzpicture}[auto, node distance=0.7cm, every loop/.style={},
                    thick,main node/.style={circle,fill=black!25,minimum size=10pt,inner sep=0pt}]
  \node[main node] (1) {};
  \node[main node] (2) [below left of=1] {};
  \node[main node] (3) [below right of=2] {};
  \node[main node] (4) [below right of=1] {};

  \path[every node/.style={minimum size=5pt}]
    (1) edge node [left] {} (4)
    (1) edge node [right] {} (3)
    (2) edge node [right] {} (3);
  \node at (2.5,-0.5) {$\Longrightarrow 010$};
\end{tikzpicture}

\vspace{0.2in}

\begin{tikzpicture}[auto, node distance=0.7cm, every loop/.style={},
                    thick,main node/.style={circle,fill=black!25,minimum size=10pt,inner sep=0pt}]
  \node[main node] (1) {};
  \node[main node] (2) [below left of=1] {};
  \node[main node] (3) [below right of=2] {};
  \node[main node] (4) [below right of=1] {};

  \path[every node/.style={minimum size=5pt}]
    (2) edge node [left] {} (3)
    (2) edge node [right] {} (4)
    (3) edge node [right] {} (4);
  \node at (2.5,-0.5) {$\Longrightarrow 011$};
\end{tikzpicture}

\end{centering}
\caption{TS code for graphical structures of size $n=4$.}
\end{figure}

\section{Preliminaries}
\label{PrelimSection}
We adopt an Erd\"{o}s-R\'enyi graph generation mechanism \cite{bollobas}, i.e., there is a simple undirected edge between pairs of nodes, independently with probability $p$. The structure of the graph is derived by removing the vertex labels. Let $\mathcal{S}(n)$ (resp. $\mathcal{G}(n)$) be the set of simple Erd\"{o}s-R\'enyi graphical structures (resp. labeled graphs) on $n$ vertices. For a structure $S\in\mathcal{S}(n)$, we define it's size as the number of vertices and $\j(S)$ captures the number of edges in $S$. When it is clear from the context we simplify $\j(S)$ as $\j$. For notational convenience, denote $m={n\choose{2}}$. Let $\pi = \pi_1,\pi_2,...,\pi_m$ be an arbitrary ordering of pairs of nodes in $G$. Construct a sequence $X_i,\: i=1,\cdots, m$ as follows: $X_i=1$, if there is an edge between pair of nodes corresponding to $\pi_i$, and 0 otherwise. Let $x_i, i=1,...,m$ denote a realization of $X_i$ for an observed $G$. Since $X_i$'s are i.i.d., let $X$ (resp. $p(X)$) be a random variable (resp. probability distribution) representing the underlying model of $X_i$'s, i.e. $p(X=1) = p(X_i=1)=p, i=1\cdots m$. Denote $\mathcal{B}$ as the class of all Bernoulli distributions over $\{0,1\}$. Let $\mathbb{E}$ and $\mathbb{V}$, denote expectation and variance with respect to $p(X)$, respectively.

We consider a source coding scheme which maps graphical structures in $\mathcal{S}(n)$ to variable length binary strings via a coding function
\begin{equation*}
\phi: \mathcal{S}(n)\rightarrow \{0,1\}^* = \{\emptyset, 0, 1, 00, 01,10,11,000, \cdot\cdot\cdot\}.
\end{equation*}

We do not impose the prefix-free condition on the coding scheme. Let $l(\phi(S))$ be the length of the compressed binary string. For any $S\in\mathcal{S}(n)$, the underlying Erd\"{o}s-R\'enyi model induces a probability distribution $\mathbb{P}_S$ on the structures within $\mathcal{S}$. We evaluate the performance of coding schemes through the $\epsilon$-coding rate for graphs of size $n$ given by

\begin{equation}
R_n(\epsilon,\phi, p)=\frac{1}{{n\choose 2}}\min\Big\{k: \mathbb{P}_S\big(\ell(\phi(S))\ge k \big)\le\epsilon\Big\}.
\end{equation}

\section{Type Size Code}
\label{TSCodeSection}
\subsection{Type Size Code}
\label{subsec::TSCSIze}
For the class of all memoryless sources over a finite alphabet, the fixed-to-variable TS code
is introduced in \cite{oliverLalithaJournal}, which sorts sequences based on the size of their elementary type classes (from
smallest to largest) and then encodes sequences to variable-length bit-strings in this order. We borrow the framework of the TS code, however, for the purpose of compressing graphical structures we define two graphical structures to be in the same type class if and only if they have the same number of edges, i.e. $T_S = \left\{S'\in\mathcal{S}(n): \j(S) = \j(S')\right\}$, where $T_S$ denotes the type class of $S$. We then sort graphical structures based on their type class sizes and map them to binary strings according to this type class size ordering. See Figure \ref{TScodeGraphs} for an example of the TS code for compressing graphical structures of size $n=4$.

\begin{Theorem}\cite{oliverLalithaJournal}
For the TS code
\begin{equation}
R_n(\epsilon,\phi,p)\leq \frac{1}{{n\choose 2}}\lceil\log{M(\epsilon)}\rceil
\end{equation}
where
\begin{equation}
\label{mEq}
M(\epsilon)=\inf_{\gamma:\mathbb{P}_S\left(\frac{1}{{n\choose 2}}\log{|T_{S}|}>\gamma\right)\leq \epsilon}\sum_{\substack{S\in\mathcal{S}(n): \\ \frac{1}{{n\choose 2}}\log{|T_{S}|}\leq \gamma} }{|T_{S}|}.
\end{equation}
\end{Theorem}

\subsection{Type Class Size}
\label{subsectionTYPECLASSSIZE}
Let $N(n,\j)$ be the number of graphs with $n$ unlabeled nodes and $\j$ simple undirected edges. Let
\begin{equation}
\Lambda(n,\j)=\frac{{{{{n}\choose{2}}}\choose{\j}}}{n!} \label{LambdaEquatidsjh}
\end{equation}
and
\begin{equation*}
\mu = \frac{2\j}{n} - \log n.
\end{equation*}
The following theorem by Wright \cite{emWright}, gives the number of graphical structures for a given number of nodes and edges.
\begin{Theorem}\label{emWrightThm}\cite{emWright}
For a constant $C_0$ independent of $n$,
$$N(n,\j) = \Lambda(n,\j)\left(1+\mathcal{O}\left(e^{-C_0\mu}\right)\right),$$
if and only if $\mu\rightarrow\infty$ as $n\rightarrow\infty$.
\end{Theorem}
Let $\mathcal{E}$ be the event where the condition of Theorem \ref{emWrightThm} is not satisfied, i.e. $\lim_{n\rightarrow\infty}\mu < \infty$. In the following lemma, we show that the necessary and sufficient condition of Theorem \ref{emWrightThm} is satisfied with high probability. The proof is an straightforward application of Chernoff bound \cite{hoeffdingW} and is provided in Appendix \ref{ChernoffAppendix}.

\begin{Lemma}
\label{ChrnofNumEdge}
There exist constants $0<\delta_1, \delta_2<1$ which are independent of $n$, such that
\begin{equation*}
\mathbb{P}_S\left(\j(S)\leq (1-\delta_1){n\choose 2}p\right) \leq e^{-{n\choose 2}\delta_2}.
\end{equation*}
\end{Lemma}

Size of the type class of $S$, $|T_S|$ is then given by
\begin{align}
\log{|T_S|}&=\log{N(n,\j)} \nonumber \\
                &=\log{{{{{{n}\choose{2}}}\choose{\j}}}}-\log{n!}+\mathcal{O}(1) \label{tsGraph}
\end{align}
where $\mathcal{O}(1)$ term is (with an abuse of notation) bounded between two positive constants independent of $n$.

Let $\j^c = m-\j$. Define \emph{empirical entropy} of the graphical structure $S$ as
\begin{equation*}
H_{\text{empirical}}(S) = -\frac{\j}{m} \log{\frac{\j}{m}} - \frac{\j^c}{m} \log{\frac{\j^c}{m}}.
\end{equation*}

The following lemma provides upper and lower bounds on the size of the graphical type class.
\begin{Lemma}
\label{TSLEmma}
With probability at least $1-e^{-{n\choose 2}\delta_2}$, we have the following upper and lower bounds for the size of the type class of a graphical structure $S\in\mathcal{S}(n)$:
\begin{equation*}
{n\choose{2}}H_{\text{empirical}}(S)-\log n! +C_L \leq \log{|T_{S}|} \leq {n\choose{2}}H_{\text{empirical}}(S)-\log {n!} + C_U
\end{equation*}
where $C_L, C_U$ are constants independent of $n$ and $\delta_2$ is the constant in Lemma \ref{ChrnofNumEdge}.
\end{Lemma}
\begin{proof}
See Appendix \ref{AppendixTSLEmma}.
\end{proof}

\section{Main Result}
\label{mainResult}
Let $H(p) = \mathbb{E}\left(\log{\frac{1}{p(X)}}\right)$ be the entropy of the underlying source generating the Erd\"{o}s-R\'enyi graph, and $\sigma^2(p)=\mathbb{V}\left(\log{\frac{1}{p(X)}}\right)$ be the varentropy of it. The following theorem provides an achievability bound for the rate of the TS code in compressing graphical structures.

\begin{Theorem}
For the TS code and any Bernoulli distribution $p\in\mathcal{B}$,
\begin{equation*}
R_n(\epsilon,\phi,p) \leq  H(p) + \frac{\sigma(p)}{\sqrt{{n\choose 2}}}Q^{-1}(\epsilon) -\frac{\log{n!}}{{n\choose 2}} + \mathcal{O}\left(\frac{1}{n^2}\right).
\end{equation*}
\end{Theorem}

\section{Proof of Theorem}
\label{ThmProof}
When it is clear from the context, we omit the underlying distribution and denote $H:=H(p)$ and $\sigma:=\sigma(p)$. For a constant $A>0$ defined in the Berry-Ess\'een Lemma \ref{berryLemma} (See Appendix \ref{Sec::beAPSec}) which is independent of $n$, let
\begin{equation}
\gamma = H + \frac{\sigma}{\sqrt{m}}Q^{-1}\left(\epsilon-\frac{A}{\sqrt{m}}-e^{-m\delta_2}\right) - \frac{\log n!}{m} +\frac{C_U}{m}
\end{equation}
where $\delta_2, C_U$ are the constants in Lemmas \ref{ChrnofNumEdge} and \ref{TSLEmma}, respectively. Denote
\begin{equation*}
p_{\gamma}  := \mathbb{P}_S\left(\log{|T_S| > m\gamma}\right).
\end{equation*}
Let $q_S(X)$ be a derived Bernoulli distribution from the structure $S$, such that $q_S\left(0\right)=\frac{\j^c}{m}$ and $q_S\left(1\right) = \frac{\j}{m}$. It is clear that
\begin{equation}
\label{entropySnEquiv}
H_{\text{empirical}}(S) = \frac{1}{m} \sum_{i=1}^{m} {-\log q_S(x_i)}.
\end{equation}
Let
\begin{equation}
\label{snEquation}
S_m = \frac{1}{\sigma(p)\sqrt{m}} \sum_{i=1}^{m} {\left(-\log{q_S(x_i)}-H(p)\right)}.
\end{equation}

Recall $\mathcal{E}$ from Subsection \ref{subsectionTYPECLASSSIZE}. We have
\begin{align}
p_{\gamma}  &= \mathbb{P}_S\left(\log{|T_S|}>m\gamma | \mathcal{E}^c\right) \mathbb{P}_G\left(\mathcal{E}^c\right) + \mathbb{P}_S\left(\log{|T_S|}>m\gamma | \mathcal{E}\right) \mathbb{P}_G\left(\mathcal{E}\right) \label{totalProb}\\
&\leq \mathbb{P}_S \left(H_{\text{empirical}}(S)  >  H(p) + \frac{\sigma(p)}{\sqrt{m}} Q^{-1}\left(\epsilon-\frac{A}{\sqrt{m}}-e^{-m\delta_2}\right)\right) + e^{-m\delta_2} \label{tooUpperB} \\
            &= \mathbb{P}\left(S_m >  Q^{-1}\left(\epsilon-\frac{A}{\sqrt{m}}-e^{-m\delta_2}\right)\right) + e^{-m\delta_2} \label{shtog} \\
            &\leq Q\left( Q^{-1}\left(\epsilon-\frac{A}{\sqrt{m}}-e^{-m\delta_2}\right)\right) +\frac{A}{\sqrt{m}} + e^{-m\delta_2} \label{Berrrr} \\
            &= \epsilon \nonumber
\end{align}
where (\ref{totalProb}) follows from the law of total probability, (\ref{tooUpperB}) follows from upper bounding $\mathbb{P}(\mathcal{E}^c)$ and $\mathbb{P}(\log{|T_G|}>m\gamma | \mathcal{E})$ by 1 in conjunction with Lemma \ref{TSLEmma}, (\ref{shtog}) is from the definitions (\ref{snEquation}) and (\ref{entropySnEquiv}), and finally (\ref{Berrrr}) is from the Berry-Esseen theorem \cite{verdulossless} (See Appendix \ref{Sec::beAPSec}). We now bound $M(\epsilon)$ using (\ref{mEq}) with this choice of $\gamma$. Let
\begin{equation}
\label{fgRep}
f(S) = H_{\text{empirical}}(S) - \frac{\log{n!}}{m} + \frac{C_U}{m}.
\end{equation}
Similarly, with an abuse of overloaded notation \footnote{The two definitions of the function $f(\cdot)$ should be distinguished based upon their arguments.}, for any Bernoulli distribution $p\in\mathcal{B}$, define
\begin{equation}
\label{fgRepII}
f(p) = H(p) - \frac{\log{n!}}{m} + \frac{C_U}{m}.
\end{equation}

The rest of the proof is similar to \cite{oliverLalithaJournal}, however, we continue the proof for completeness.

\begin{Lemma} \cite{oliverLalithaJournal}
\label{fIsLipschitz}
There exists a Lipschitz constant $K_0$ independent of $n$, such that for any two Bernoulli distributions $p, \tilde{p}\in\mathcal{B}$,
\begin{equation}
|f(p)-f(\tilde{p})|\leq K_0 \|p-\tilde{p}\|.
\end{equation}
\end{Lemma}

Fixing $\Delta=\frac{1}{m}$, we have
\begin{align}
M(\epsilon) &\leq \sum_{\substack{S\in\mathcal{S}(n):\\ \frac{1}{m}\log{|T_{S}|}\leq \gamma}}{|T_{S}|} \nonumber \\
 &\leq \sum_{\substack{S\in\mathcal{S}(n):\\ f(S) -\frac{C_d}{m}\leq \gamma}}{2^{mf(S)}} \nonumber \\
            &= \sum_{i=0}^{\infty}\sum_{\substack{S\in\mathcal{S}(n):\\ f(S)\in\mathcal{A}_i }} {2^{mf(S)}} \nonumber \\
						&\leq \sum_{i=0}^{\infty}\left|\left\{S\in\mathcal{S}(n):f(S)\in\mathcal{A}_i\right\}\right| \cdot  2^{m\gamma+C_d - mi\Delta} \label{mEpsilon}
\end{align}
where $C_d=C_U-C_L$ and $\mathcal{A}_i = (\gamma +\frac{C_d}{m}-(i+1)\Delta, \gamma +\frac{C_d}{m}-i\Delta]$. For a Bernoulli distribution $p\in\mathcal{B}$, define its 2-norm ball of radius $r$ as $B_{r}(p) = \{p'\in \mathcal{B}: \|p-p'\|\leq r\}$. By extension, for a graphical structure $S\in\mathcal{S}(n)$, define its 2-norm ball of radius $r$ as $B_r(S) := B_r(q_S)$, where $q_S$ is the derived empirical distribution of $S$ as defined at the beginning of this section. Note that for any two different structures $S_1, S_2 \in \mathcal{S}(n)$, $B_{\frac{1}{2m}}(S_1)$ and $B_{\frac{1}{2m}}(S_2)$ are disjoint. Moreover, observe that $\text{Vol}\left(B_{\frac{1}{2m}}(S)\right) = \frac{1}{m}$.
We have
\begin{align}
\Big|\{S\in\mathcal{S}(n):f(S)\in\mathcal{A}_i\}\Big|&=\sum_{\substack{S\in\mathcal{S}(n):\\f(S)\in\mathcal{A}_i}}\frac{\text{Vol}\left(B_{\frac{1}{2m}}(S)\right)}{\frac{1}{m}} \nonumber \\
&= m\sum_{\substack{S\in\mathcal{S}(n):\\f(S)\in\mathcal{A}_i}}{\text{Vol}\left(B_{\frac{1}{2m}}(S)\right)} \nonumber \\
&= m{\text{Vol}\left(\bigcup_{\substack{S\in\mathcal{S}(n):\\f(S)\in\mathcal{A}_i}}B_{\frac{1}{2m}}(S)\right)} \label{disjointnessEqw} \\
&\leq m{\text{Vol}\left(\bigcup_{\substack{p\in\mathcal{B}:\\f(p)\in\mathcal{A}_i}}B_{\frac{1}{2m}}(p)\right)} \label{contiskdjh}
\end{align}
where (\ref{disjointnessEqw} is from disjointness of the balls. Let $\rho(\lambda)=\text{Vol}\{p\in\mathcal{B}:f(p)\leq \lambda\}$. The following lemma from \cite{oliverLalithaJournal}, shows the Lipschitzness of $\rho(\cdot)$.
\begin{Lemma}\cite{oliverLalithaJournal}
\label{rhoIsLipschitz}
There exists a Lipschitz constant $K_1$ such that for all $a,b$,
\begin{equation*}
|\rho(a)-\rho(b)|\leq K_1|a-b|.
\end{equation*}
\end{Lemma}

We continue from (\ref{contiskdjh}). Let $a = \gamma +\frac{C_d}{m}-(i+1)\Delta$. We have
\begin{align}
\Big|\Big\{S\in\mathcal{S}(n):f(S)\in \mathcal{A}_i\Big\}\Big| &\leq m\cdot \text{Vol}\left(\bigcup_{ a<f(p)\leq a+\Delta}{B_{\frac{1}{2m}}(p)}\right) \nonumber \\
&\leq m\text{Vol}\left(\Big\{p:f(p)\in \left(a- \frac{K_0}{2m},a+\Delta+ \frac{K_0}{2m}\right)\Big\}\right) \label{useLipschoys} \\
&= m\left(\rho\left(a+\Delta+\frac{K_0}{2m}\right)-\rho\left(a-\frac{K_0}{2m}\right)\right) \nonumber \\
&\leq mK_1K_0\cdot \left(\Delta+\frac{K_0}{m}\right) \label{sizeTau}
\end{align}
where (\ref{useLipschoys}) is from the observation that for any $\tilde{p}\in B_{\frac{1}{2m}}(p)$, $|f(\tilde{p})-f(p)|\leq \frac{K_0}{2m}$ and (\ref{sizeTau}) is from Lemma \ref{rhoIsLipschitz}. Applying (\ref{sizeTau}) to (\ref{mEpsilon}), we obtain
\begin{align*}
M(\epsilon) &\leq \sum_{i=0}^{\infty}{ mK_1K_0\cdot \left(\Delta+\frac{K_0}{m}\right)\cdot 2^{m\gamma +C_d-mi\Delta}} \nonumber \\
            &= mK_1K_0\cdot \left(\Delta+\frac{K_0}{m}\right)\cdot 2^{m\gamma+C_d}\cdot \frac{1}{1-2^{-m\Delta}}.
\end{align*}
Since $\Delta=\frac{1}{m}$, we have
\begin{align*}
\log {M(\epsilon)}
					&\leq \log\left(K_1K_0(K_0+1)\right) +m\gamma  + C_d + 1 \nonumber \\				&=mH(p)+\sigma\sqrt{m}Q^{-1}\left(\epsilon-\frac{A}{\sqrt{m}}-e^{-m\delta_2}\right) -\log{n!}+ C_U +C_1 \\
									&\leq mH(p)+\sigma\sqrt{m}Q^{-1}(\epsilon)-\log{n!}+C
\end{align*}
for constants $C_1,C$ independent of $n$.

\section{Conclusion and Future Work}
\label{SecConclusion}
We proposed a variant of the Type Size code for compressing graphical structures. Erd\"{o}s-R\'enyi model is adopted as the underlying mechanism for generating the random graphs. We provided an analysis to derive the fine asymptotics of the overflow rate of the proposed Type Size code for compressing such structures. However, Erd\"{o}s-R\'enyi model fails to fully represent the real-world networks such as the world wide web. The alternative models include the power law and the preferential attachment models. We study the finite blocklength compression of graphical structures not generated by the Erd\"{o}s-R\'enyi model as a future work. The multigraph version of the problem, which permits multiple edges between pairs of nodes and the lossy version of the problem are also interesting future directions of this research.
{\bibliographystyle{IEEEtran}
\bibliography{reputation}}

\begin{appendices}

\section{Proof of Lemma \ref{ChrnofNumEdge}}
\label{ChernoffAppendix}
Observe that the number of edges in the Erd\"{o}s-R\'enyi graph can be derived as the sum of $m$ i.i.d. random variables $X_i$, defined in Subsection \ref{subsec::TSCSIze}, i.e. $\j(S) = \sum_{i=1}^{m} X_i$. For positive $\alpha,t>0$, from Chernoff bound \cite{hoeffdingW}, we have
\begin{equation}
\mathbb{P}_S\left(\j(S)\leq \alpha\right) \leq \min_{t>0} e^{t\alpha} \prod_{i}{\mathbb{E}\left(e^{-tX_i}\right)}.
\end{equation}
On the other hand
\begin{align}
\mathbb{E}\left(e^{-tX_i}\right) &= pe^{-t} + (1-p) \nonumber \\
                      &= 1+ p\left(e^{-t}-1\right) \nonumber \\
                      &\leq e^{p(e^{-t}-1)} \label{lnxXminOne}
\end{align}
where (\ref{lnxXminOne}) follows from $\ln x\leq x-1$, for any positive $x>0$. Hence
\begin{equation}
\mathbb{P}_S(\j(S)\leq \alpha) \leq \min_{t>0} e^{t\alpha}\cdot e^{mp\left(e^{-t}-1\right)}. \label{sincEThsi}
\end{equation}
Since (\ref{sincEThsi}) holds for any $\alpha,t>0$, therefore, we may take $\alpha=(1-\delta_1)mp$ and $t= -\ln(1-\delta_1)>0$, for an arbitrary $0< \delta_1 < 1$. Subsequently, we obtain
\begin{align}
\mathbb{P}_S\left(\j(S)\leq (1-\delta_1)mp\right) &\leq \left(\frac{1}{1-\delta_1}\right)^{\left(1-\delta_1\right)mp} e^{-mp\delta_1} \nonumber \\
                                 &= e^{-m\delta_2}
\end{align}
where $\delta_2 = -\delta_1 - (1-\delta_1)\ln(1-\delta_1)>0$ is a positive constant.

\section{Proof of Lemma \ref{TSLEmma}}
\label{AppendixTSLEmma}
Utilizing Theorem \ref{emWrightThm}, we proceed by providing tight upper and lower bounds for $\Lambda(n,\j)$ in (\ref{LambdaEquatidsjh}).
\textbf{Upper Bound}:\\
Recall ${\j}^c={n\choose 2}-\j$. Using the Stirling's formula \cite{imrebook}, we have
\begin{align}
\log{{n\choose2}\choose{\j}}&=\log\frac{{n\choose2} !}{\j!\left({n\choose 2} -\j\right)!} \nonumber \\
                       &\leq \log\frac{\sqrt{2\pi} {n\choose 2} ^{{n\choose 2}+\frac{1}{2}}2^{-{n\choose 2}+\frac{1}{12{n\choose 2}}}} {\left(\sqrt{2\pi}{\j}^{\j+\frac{1}{2}}2^{-\j+\frac{1}{12(\j+1)}}\right)\left(\sqrt{2\pi}{\j^c}^{\j^c+\frac{1}{2}}e^{-\j^c+\frac{1}{12(\j^c+1)}}\right)} \nonumber \\
											&= \left({n\choose 2}+\frac{1}{2}\right)\log{{n\choose 2}}-{n\choose 2}+\frac{1}{12{n\choose 2}}-\left(\j+\frac{1}{2}\right)\log \j + \j -\frac{1}{12(\j+1)} -\left({\j}^c+\frac{1}{2}\right)\log{{\j}^c} + {\j}^c - \frac{1}{12 ({\j}^c+1)} - \log{\sqrt{2\pi}}\nonumber \\
											&\leq\left({n\choose 2}+\frac{1}{2}\right)\log{{n\choose 2}}-\left(\j\log{\j}+{\j}^c\log {\j}^c\right)-\frac{1}{2}\left(\log{\j}+\log {\j}^c\right)+\frac{1}{12}-\log{\sqrt{2\pi}}.\label{firstPartosUpperBound}
\end{align}
Recall $q_S\in\mathcal{B}$, a derived Bernoulli distribution with $q_S\left(X=0\right)=\frac{{\j}^c}{{n\choose 2}}$ and $q_S\left(X=1\right)=\frac{{\j}}{{n\choose 2}}$, with an entropy $H\left(q_S\right)=H_{\text{empirical}}(S)$. Note that
\begin{align}
-{\j}\log{\j}-{\j}^c\log{{\j}^c} &= -{n\choose 2}\frac{{\j}}{{n\choose 2}}\log{\frac{{\j}}{{n\choose 2}}}-{n\choose 2}\frac{{\j}^c}{{n\choose 2}}\log{\frac{{\j}^c}{{n\choose 2}}} -{\j}\log{{{n\choose 2}}}-{\j}^c\log{{{n\choose 2}}} \nonumber \\
											&={n\choose 2}H_{\text{empirical}}(S)-{n\choose 2}\log{{n\choose 2}} \label{secondPatOfUper}.
\end{align}
Without loss of generality, we assume $0<\j,{\j}^c<{n\choose 2}$. This can be accommodated by adding two extra header bits to the compressed data stream, to indicate if the graph is empty, complete, or else. Hence, we have
\begin{equation}
\label{mBoundjjc}
\log{\j} + \log{\j^c} \geq \log{{n\choose 2}}-1.
\end{equation}
(\ref{firstPartosUpperBound}) in conjunction with (\ref{secondPatOfUper}, \ref{mBoundjjc}), gives the upper bound.\\
\textbf{Lower Bound}: \\
Using the Stirling's formula \cite{imrebook}, we have
\begin{align*}
\log{{n\choose2}\choose{{\j}}}&=\log\frac{{n\choose2} !}{{\j}!\left({n\choose 2} -{\j}\right)!} \nonumber \\
                       &\geq \log\frac{\sqrt{2\pi} {n\choose 2} ^{{n\choose 2}+\frac{1}{2}}2^{-{n\choose 2}+\frac{1}{12\left({n\choose 2}+1\right)}}} {\left(\sqrt{2\pi}{\j}^{\j+\frac{1}{2}}2^{-\j+\frac{1}{12\j}}\right)\left(\sqrt{2\pi}{{\j}^c}^{{\j}^c+\frac{1}{2}}2^{-{\j}^c+\frac{1}{12{\j}^c}}\right)} \nonumber \\
											&= \left({n\choose 2}+\frac{1}{2}\right)\log{{n\choose 2}}-{n\choose 2} + \frac{1}{12\left({n\choose 2} + 1\right)}-\left({\j}+\frac{1}{2}\right)\log{\j}+\j-\frac{1}{12\j}-\left({\j}^c+\frac{1}{2}\right)\log{{\j}^c}+{\j}^c-\frac{1}{12{\j}^c} -\log{\sqrt{2\pi}}\nonumber \\
											&=\left({n\choose 2}+\frac{1}{2}\right)\log{{n\choose 2}}-\left({\j}\log{\j}+{\j}^c\log {\j}^c\right)-\frac{1}{2}\left(\log{\j}+\log{\j}^c\right)-\frac{1}{6}-\log{\sqrt{2\pi}}.
\end{align*}
Moreover, note that
\begin{equation}
\log{\j} + \log{{\j}^c} \leq \frac{1}{2}\log{{n\choose 2}}.
\end{equation}
The rest of the proof is similar to the proof of the upper bound.

\section{Berry-Ess\'een Bound}
\label{Sec::beAPSec}
\begin{Lemma}\label{berryLemma}\cite{verdulossless}
Let $\{Z_i\}$ be independent and identically distributed random variables with zero mean and unit variance, and let $\tilde{Z}$ be a standard normal. Then, for a constant $A>0$ independent of $n$, all $n\geq 1$ and all $z$ we have
\begin{equation*}
\left|\mathbb{P}\left(\frac{1}{\sqrt{n}}\sum_{i=1}^{n}Z_i\leq z\right) - \mathbb{P}\left(\tilde{Z}\leq z\right)\right|\leq \frac{A}{\sqrt{n}}
\end{equation*}

\end{Lemma}

\end{appendices}

\end{document}